\documentclass[12pt]{amsart}

\usepackage{graphicx}
\usepackage{amsmath,amsthm}
\makeatletter
\def\amsbb{\use@mathgroup \M@U \symAMSb}
\makeatother
\usepackage{bbold}
\usepackage[all]{xy}
\usepackage{amssymb}
\usepackage{caption}
\usepackage[latin1]{inputenc}
\usepackage{enumerate}
\usepackage{mathrsfs}  
\usepackage{dsfont}
\usepackage[bbgreekl]{mathbbol}
\usepackage{geometry}
\usepackage{datetime}

\theoremstyle{plain}

\newtheorem{thm}{Theorem}[section]

\newtheorem{lemma}{Lemma}[section]
\newtheorem{prop}{Proposition}[section]
\newtheorem{defi}{Definition}[section]

\newtheorem{obs}{Remark}[section]

\newcommand{\cC}{{\mathcal C}}
\numberwithin{equation}{section}

\def\C{\mathcal {C}}

\def\A{{\mathcal{C}}^\infty(Q)} 
\def\AT2{{\mathcal{C}}^\infty(Q_k^2)}

\def\U{\mathcal {U}}

\def\R{\amsbb{R}}
\def\m{\mathfrak m}

\def\To{\longrightarrow}

\def\Tau{\mathcal{T}}

\def\be{\begin{equation}
}

\def\ee{\end{equation}
}

\def\iiota{\iota}

\begin{document}

\title[Newton's second law in field theory]
 {Newton's second law in field theory}

\author{ R. J.  Alonso-Blanco and J.  Mu\~{n}oz-D{\'\i}az}

\address{Departamento de Matem\'{a}ticas, Universidad de Salamanca, Plaza de la Merced 1-4, E-37008 Salamanca,  Spain.}
\email{ricardo@usal.es,clint@usal.es}

\begin{abstract}
In this article we present a natural generalization of Newton's Second Law valid in field theory, i.e., when the parameterized curves are replaced by parameterized submanifolds of higher dimension. For it we introduce what we have called the geodesic $k$-vector field, analogous to the ordinary geodesic field and which describes the inertial motions (i.e., evolution in the absence of forces). From this generalized Newton's law, the corresponding Hamilton's canonical equations of field theory (Hamilton-De Donder-Weyl equations) are obtained by a simple procedure. It is shown that solutions of generalized Newton's equation also hold the canonical equations. However, unlike the ordinary case, Newton equations determined by different forces can define equal Hamilton's equations. 
\end{abstract}
\bigskip

\maketitle

\setcounter{tocdepth}{1}


\tableofcontents

\section{Introduction}\label{intro}

At the core of Classical Mechanics is Newton's second law: force equals mass times acceleration. Later elaborations led to the equivalent and successful  Lagrangian and Hamiltonian formulations of mechanics.  The passage from ordinary mechanics to field theory is done precisely in Lagrangian and Hamiltonian terms. Therefore, a natural question arises: is there a version of Newton's second law in field theory? We give an affirmative answer in this article.
\medskip

 Let us recall the basic points of the ordinary version. Newton's second law says, essentially, that the motions are governed by second order differential equations (a specific type of vector fields defined on the tangent bundle of the configuration manifold $Q$). When $Q$ is endowed with a pseudoriemannian metric, we can set which are the inertial motions (integral curves of the geodesic field). In addition, the metric gives us an isomorphism of the tangent bundle with the cotangent bundle, where the Liouville form is defined. With the above ingredients (a tangent field describing inertial motion and a symplectic structure on $TQ$) we can associate a second order differential equation to each given \emph{force} (see \cite{MecanicaMunoz,RM}).
 \medskip

 When dealing with field theory, the curves describing the motions are replaced by submanifolds of dimension $k$ greater than one. Hence, in order to get the generalized Newton's law, we have to replace the tangent bundle $TQ$ with an analogous space for $k>1$: the bundle of $(k, 1)$-velocities, $Q_k^1$, a particular case of Weil bundle (see \cite{Weil, KolarNatural, MMR1}). It is also necessary to consider second order partial differential equations and, in addition,  to obtain a generalized geodesic field to be able to define the \emph{inertial motion}: we will see that this is possible by using the metric. Also by means of the metric, we obtain an isomorphism of $Q_k^1$ to its dual (see below), where the generalized Liouville form is defined (\emph{polysimplectic structure} \cite{Gunther}). Once this is done, the rest goes parallel to the ordinary case.
 \medskip

 This generalized Newton's second law is an equation with values in certain linear endomorphisms space. When we take the trace, we recover the Euler-Lagrange and the canonical Hamiltonian formalisms (Hamilton-De Donder-Weyl equations). As a consequence, all the solutions of the Newton's equation also hold the associated Euler-Lagrange and canonical equations. However, in general, the converse is not true. Hence, there is an essential difference with the classical case: now ($k>1$), Newton's law is no longer equivalent to the Euler-Lagrange equations or the canonical equations. In that sense, we can say that Newton's equation  distinguishes between systems that, from the Lagrangian and Hamiltonian points of view, look identical.
 \medskip

To develop this program, we will follow \cite{MecanicaMunoz} (see also \cite{RM}) where has been given a presentation of Mechanics allowing several conceptual advances (v.g., a clear understanding of what \emph{absolute time} \cite{TiempoMunozAlonso} and \emph{relativistic forces} \cite{RelatividadMunoz,RM} are).
\medskip

The framework used in this paper is close to that of Günther \cite{Gunther}, who developed the so-called \emph{polysymplectic formalism}, which has been extended  in several directions since then (see for instance \cite{Awane, LecandaSS, RRS, deLeon} and references therein; it should be clarified that there are other different approaches to field theory in their Lagrangian and Hamiltonian forms; see, v.g., \cite{Rund,GPerez,GS})).
\medskip

The structure of this paper is the following: in Section \ref{S:velocities} they are introduced the bundles of $(k,1)$ and $(k,2)$-velocities;
  Section~\ref{S:geodesic} presents the definition of second order PDE and, a particular case,  the \emph{geodesic $k$-field}, is constructed;
    in Section~\ref{S:Liouville} the dual spaces are considered: bundles of \emph{covelocities} joint with the \emph{generalized Liouville form}; with all the previous work done, in Section~\ref{S:Newton} the \emph{Newton's second law for fields} is obtained; by taking the trace  we get the Euler-Lagrange and Hamilton canonical equations in Section~\ref{S:Lagrange} and, also, we obtain their variational properties.
  \bigskip

\noindent{\bf Notation and conventions.} Along the paper $Q$ will be an $n$-dimensional $\C^\infty$ manifold. Einstein convection on repeated indexes will be assumed. For an smooth map $\varphi$ between manifolds, $\varphi_*$ will denote its tangent map (at a given point), and $\varphi^*$ will be the associated pull-back of differential forms. The inner contraction of a vector field $D$ with a differential form $\omega$ is denoted by $\iiota_{D}\omega$.

\bigskip

\section{Newton's second law}
In order to present our goal and also to fix notation, it will be useful to recall the model of the Classical Mechanics. We will follow the presentation in \cite{MecanicaMunoz,RM}.
\medskip

Let $Q$ be an $n$-dimensional smooth manifold, and $\pi\colon TQ\to Q$ be its
tangent bundle.

There is a canonical derivation $$\dot d\colon\A\to\C^\infty(TQ),\quad\text{ defined by}\quad f\mapsto\dot df:=\dot f,$$ where
$\dot f$ is $df$ considered as a function on $TQ$; that is to say,
$$\dot f(v_q):=v_q(f),\quad\forall f\in\A, v_q\in TQ.$$

When $\{q^1,\dots,q^n\}$  is a coordinate system on $\U\subset Q$, the set $\{q^1,\dots,q^n,\dot q^1,\dots,\dot q^n\}$ defines a local chart on $\pi^{-1}(\U)\subseteq TQ$. In these coordinates,
 $$\dot d=\dot q^i\frac{\partial}{\partial q^i}.$$

\begin{defi} \textbf{\em (Second Order Differential Equation).}\label{def: second_order_equation}
A vector field $D$ on $TQ$ is a \emph{second order differential
equation} when its restriction (as derivation) to the subring
$\cC^\infty(Q)$ of $\cC^\infty(TQ)$ is $\dot d$. This is
equivalent to have $\pi_*(D_{v_q})=v_q$ at each point $v_q\in
T_qQ$.
\end{defi}

Locally, a second order differential equation is
\begin{equation*}
D=\dot q^i\frac{\partial}{\partial q^i}+f^i(q,\dot
q)\frac{\partial}{\partial \dot q^i}\\
\end{equation*}
for suitable functions $f^i$ on $TQ$.

A curve $q^i=q^i(\gamma(t))$, $\dot q^i=\dot q^i(\gamma(t))$ in $TQ$ is integral for such a vector field $D$ if and only if
$dq^i(\gamma(t))/dt=\dot q^i(\gamma(t))$ and $d\dot q^i(\gamma(t))/dt=f^i(q^j(\gamma(t)),dq^j(\gamma(t))/dt)$. That is to say, if it holds the following system of second order ordinary differential equations
$$\frac{d^2q^i(\gamma(t))}{dt^2}=f^i(q^j(\gamma(t)),dq^j(\gamma(t))/dt),\quad i=1,\dots,n.$$

Let $T^*Q$ be the cotangent bundle of $Q$ and $\overline\pi\colon T^*Q\to Q$
the canonical projection. Recall that the \emph{Liouville form}
$\theta$ on $T^*Q$ is defined by $\theta_{\sigma_q}=\overline\pi^*(\sigma_q)$
for each $\sigma_q\in T^*_qQ$.
\medskip

The $2$-form $d\theta$ is the natural symplectic form
associated to $T^* Q$. In local coordinates $\{q^1,\dots, q^n\}$ in
$Q$, and the corresponding $\{q^1,\dots, q^n,p_1,\dots, p_n\}$ for
$T^*Q$ (defined by $p_i(\sigma_q):=\sigma_q(\partial/\partial q^i)$), we have
\begin{equation*}
\theta=p_idq^i,\qquad d\theta=dp_i\wedge dq^i \ .
\end{equation*}

Let $g$ be a pseudoriemannian metric in $Q$.
Then we have an isomorphism of vector fiber bundles $TQ\simeq T^*Q$, $v_q \mapsto \iiota_{v_q}g$
($\iiota_{v_q}g$ is the inner contraction of $v_q$ with $g$). Using
the above isomorphism we can transport all the structures of
$T^*Q$ to $TQ$. In particular, we will work with the Liouville form $\theta$
and the symplectic form $d\theta$ transported to $TQ$ with the
same notation.

 If the expression in
local coordinates of the metric is $g=g_{ji}\ dq^jdq^i$,
then the local equations for the isomorphism $TQ\approx T^*Q$ are
\begin{equation*}\label{isomorfismo}
p_j=g_{ji}\ \dot q^i,
\end{equation*}
and the Liouville form on $TQ$ is given by
\begin{equation*}\label{Liouvilletangente}
\theta=g_{ji}\ \dot q^idq^j.
\end{equation*}

\begin{defi}\textbf{\em (Kinetic Energy).}\label{cinetica}
The function $T$ on $TQ$ defined by $T(v_q):=\frac 12 g(v_q,v_q)$, $v_q\in TQ$,  is the \emph{kinetic
energy} associated to the metric $g$. It holds, equivalently,  $T=\frac 12\,g(\dot d,\dot d)$.
\end{defi}

Given a tangent vector $v_q\in T_qQ$, there is a  geodesic curve on $Q$ such that its initial velocity at $q$ is exactly $v_q$. As a consequence, trough each point $v_q\in TQ$ passes the prolongation of a unique geodesic curve. In this way, a vector field $D_G$ on $TQ$ is defined. In coordinates,
\begin{equation}\label{geodesico}
D_G=\dot q^i\frac{\partial}{\partial q^i}-\Gamma_{ij}^\ell \dot q^i
\dot q^j \frac{\partial}{\partial \dot q^\ell}
\end{equation}

In fact, $D_G$ can be characterized in the following terms:
\begin{lemma} 
The geodesic field $D_G$ of the metric $g$ is the second order
differential equation determined by
\begin{equation}\label{ecuaciongeodesica}
\iiota_{D_G}d\theta+dT=0 \ .
\end{equation}
\end{lemma}
 The geodesic field $D_G$ is chosen as the
origin in the affine bundle of second order differential
equations. With this choice we establish a one-to-one
correspondence between second order differential equations $D$ and
vertical tangent fields V.
$$
D\longleftrightarrow V:=D-D_G.
$$

On the other hand, a differential 1-form on $TQ$ is said to be \emph{horizontal} if vanishes when is applied to a vertical vector (i.e., a vector in the kernel of $\pi_*$, where $\pi\colon TQ\to Q$ is the canonical projection). In particular, Liouville form $\theta$ is horizontal.
\medskip

Then, it can be proved the following main result
\begin{thm}[\cite{MecanicaMunoz}]\label{teoremaalfa}
The metric $g$ establishes a one-to-one correspondence between
second order differential equations and horizontal $1$-forms in
$TQ$.
The second order differential equation $D$ and the horizontal
$1$-form $F$ that correspond to each other are related by
\begin{equation}\label{formulaalpha}
\iiota_Dd\theta+dT=F\ .
\end{equation}
\end{thm}
\medskip

The correspondence (\ref{formulaalpha}) will be called the \emph{
Newton's Second Law}: \emph{under the influence of a force $F$ the
trajectories of the mechanical system satisfy the second order differential equation $D$ associated with $F$}. Indeed, $F$ is a force, $V=D-D_G$ will be the acceleration (relative to the \emph{inertial motion} $D_G$) and $\iiota_Vd\theta=F$ relates both of them.

\bigskip

\section{Bundle of contact elements. $k$-Velocities}\label{S:velocities}

The points of the manifold $Q$ correspond to morphisms of algebras $\A\to\R$. If we substitute $\R$ by the  \emph{dual numbers} $\R[\epsilon]/\epsilon^2$ we get the tangent vectors as a kind of generalized ``points'' of $Q$. That is the idea behind the theory developed by Weil in his paper on `points proches' or \emph{near elements} \cite{Weil} (see also \cite{MMR1,KolarNatural}). In that theory they are considered morphisms of the ring of smooth functions of a given manifold to a \emph{Weil algebra}, i.e., a finite dimensional local and rational $\R$-algebra. The main examples of Weil algebras are the quotients
\begin{equation}\label{Rkl}
\R_k^\ell:=\R[\epsilon^1,\dots,\epsilon^k]/(\epsilon^1,\dots,\epsilon^k)^{\ell+1},
\end{equation}
where $\R[\epsilon^1,\dots,\epsilon^k]$ are the polynomials in $k$ undetermined $\epsilon^i$, and $(\epsilon^1,\dots,\epsilon^k)$ is the maximal ideal generated by them (so that (\ref{Rkl}) is the ring of polynomials truncated at the degree $\ell$).

\begin{defi}\label{contactelements}
 A \emph{$(k,1)$-contact element} or \emph{$(k,1)$-velocity} $X_q$ at $q\in Q$ is one of the following equivalent objects:
\begin{enumerate}
\item A $k$ tuple $X_1,\dots X_k$ of vectors in $T_qQ$
\item A linear map $X_q\colon\R^k\to T_qQ$, or $X_q=X_\alpha\otimes e^\alpha$, where $X_\alpha\in T_qQ$ and $\{e^\alpha\}$ is the canonical  basis of $(\R^k)^*$.
\item The tangent map $d_t\gamma$ (or, equivalently, its dual $(d_t\gamma)^*$, the cotangent map) associated with an smooth map $\gamma\colon \U\to Q$, where $\U$ is a neighborhood of $t\in\R^k$ and $\gamma(t)=q$.
\item An $\R$-algebra morphism $\C^\infty(Q)\to \R_k^1$, which when composed with the canonical $\R_k^1\to\R$ is the map $f\mapsto f(q)$, $f\in\C^\infty(Q)$.
\end{enumerate}
The set of $(k,1)$-contact elements of $Q$ is denoted by $Q_k^1$. As a particular case, $Q_1^1=TQ$.
\end{defi}

The equivalence of the above items 1, 2 and 3 is obvious because $T_t\,\U=\R^k$. Item 4 follows from 3 by taking $\gamma^*\colon\C^\infty(Q)\to\C^\infty(\U)$ followed be the quotient map
$$\A\to\C^\infty(\U)/\m^2_t\simeq \R_k^1,$$
where $\m_t$ is the maximal ideal $\{f\in\C^\infty(\U)\,|\,f(t)=0\}$.

Let us identify $X_q\in Q_k^1$ with a linear map $\R^k\to T_qQ$; if $\{q^j\}_{j=1,\dots,n}$ are local coordinates around $q\in Q$, we have
  $$X_q( t^1,\dots, t^k)=  t^\alpha \dot q_\alpha^j\left(\frac{\partial}{\partial q^j}\right)_q,\qquad \forall( t^1,\dots, t^k)\in\R^k,$$
  for suitable scalars $\dot q_\alpha^j=\dot q_\alpha^j(X_q)$. Equivalently,  $X_q$ considered as a morphism $\A\to\R_k^1$ is given by
$q^j\mapsto q^j(q)+\dot q^j_\alpha(X_q)\epsilon^\alpha$.
The set $\{q^j,\dot q_\alpha^j\}$ defines a local chart on $Q_k^1$; so $Q_k^1$ is endowed with a smooth manifold structure. With this structure, the natural map $\pi\colon Q_k^1\to Q$ is a $\C^\infty$ regular projection. 
\medskip

 For each $\alpha=1,\dots,k,$ we define the derivative:
\begin{equation}\label{puntoderivada}
\dot\partial_\alpha\colon\C^\infty(Q)\to \C^\infty(Q_k^1),\qquad \dot\partial_\alpha(f)(X_q):=X_\alpha f,
\end{equation}
for all $f\in\A$ and $X_q=(X_1,\dots,X_k)\in Q_k^1$. In local coordinates, $$\dot\partial_\alpha=\dot q_\alpha^i\,\frac{\partial}{\partial q^i}.$$

 \subsection{First order prolongation of a parameterized submanifold}
 Let $\U$ be an open subset of $\R^k$ with coordinates $t^1,\dots,t^k$ (associated with the canonical basis $\{e_\alpha\}$). Each smooth map $\gamma\colon \U\to Q$ can be prolonged to the manifold of $(k,1)$ contact elements in the natural way:
 \begin{equation}\label{prolongation1}
 \gamma^1\colon \U\to Q_k^1,\quad t\mapsto \gamma^1(t):=(\gamma_*)_t,
 \end{equation}
 or, within the obvious identifications, $\gamma^1(t)=(\gamma_*)_t(\partial/\partial {t^1},\dots,\partial/\partial t^k)_t$. So that, locally
 $\gamma^1$ is written as
 \begin{equation}\label{prolongaciongamma}
 q^i=q^i(\gamma(t)),\qquad \dot q^i_\alpha=\frac{\partial q^i(\gamma(t))}{\partial t^\alpha}.
 \end{equation}

 \subsection{First order differential equation defined by a field of $k$-velocities}
 Let $X\colon Q\to Q_k^1$ a section of the natural projection $Q_k^1\to Q$. This is equivalent to a collection of $k$ vector fields $(X_1,\dots,X_k)$ on $Q$. In local coordinates
 $$X_\alpha=f_\alpha^i(q)\frac{\partial}{\partial q^i},$$
 for certain functions $f^i_\alpha\in\A$.
 From $X$ is easy to define a system of first order partial differential equations:
 \begin{defi}\label{FOPDE}
 Let $\U$ be an open subset of $\R^k$ (coordinates $t^1,\dots,t^k$). A smooth map $\gamma\colon \U\to Q$ is a solution of the \emph{system of first order partial differential equations defined by $X$} if, for each $t\in \U$, it holds $\gamma^1(t)=X_{\gamma(t)}$ (i.e., $\gamma^1=X\circ\gamma$).
 \end{defi}

 By using (\ref{prolongaciongamma}), in local coordinates, $\gamma$ is a solution of $X$ if
 \begin{equation}\label{FOPDElocal}
 \frac{\partial q^i(\gamma(t))}{\partial t^\alpha}=f_\alpha^i(q^1(\gamma(t)),\dots,q^k(\gamma(t))),\quad \alpha=1,\dots,k.
 \end{equation}

 When $k=1$, Definition \ref{FOPDE} matches the well-known differential equation determined by a tangent vector field.

 \subsection{Second order contact elements}
 Item 4 in Definition \ref{contactelements} is the most easily generalizable: it is sufficient to substitute $\R_k^1$ by other suitable Weil algebra. For example, we can define second order contact elements:
 \begin{defi}
 A  \emph{$(k,2)$-contact element} or \emph{$(k,2)$-velocity} $S_q$ at $q\in Q$ is an $\R$-algebra morphism $\C^\infty(Q)\to \R_k^2$, which when composed with $\R_k^2\to\R$ is the map $f\mapsto f(q)$, $f\in\C^\infty(Q)$. In local coordinates such a map is described by
 $$q^i\mapsto q^i(q)+\dot q^i_\alpha(S_q)\epsilon^\alpha+\ddot q^i_{\alpha\beta}(S_q)\,\epsilon^\alpha\epsilon^\beta,$$
 for appropriate scalars  $\dot q^i_\alpha(S_q)$, $\ddot q^i_{\alpha\beta}(S_q)$ (symmetric in $\alpha\beta$). The functions $q^i,\dot q^i_\alpha,\ddot q^i_{\alpha\beta}$ so defined, serve as a local chart. In this way, $Q_k^2$  is endowed with a smooth manifold structure. There is a natural projection $Q_k^2\to Q_k^1$, by composing with $\R^2_k\to\R^1_k$ and is given in coordinates by $(q^i,\dot q^i_\alpha,\ddot q^i_{\alpha\beta})\mapsto (q^i,\dot q^i_\alpha)$.
 \end{defi}

 In fact, given an arbitrary Weil algebra $A$ is possible to define the manifold of $A$-near points $Q^A$ \cite{Weil} (also known as Weil bundle and denoted by  $T^AQ$ \cite{KolarNatural}, or bundle of $A$-contact elements in \cite{MMR1} in a terminology closer to that employed by S. Lie). For that, it is sufficient  to consider morphisms $\A\to A$, etc.

 It is possible, anyway, to relate this algebraic definition with infinitesimal calculus as follows: consider an smooth map
 $\gamma\colon \U\to Q$ with $\gamma(t)=q$ (as in Definition \ref{contactelements}). By transposition we get the associated map between rings of functions $\gamma^*\colon\C^\infty(Q)\to\C^\infty(\U)$. Next, we compose with the quotient by $\m_t^3$ (or, that is the same, we take the second order Taylor expansion):
 $$\A\to\C^\infty(\U)/\m_t^3\simeq \R_k^2,$$
which is, by definition, the  $(k,2)$-contact element with:
$$\dot q^i_\alpha(S_q)= \left(\frac{\partial q^i(\gamma(t))}{\partial t^\alpha}\right)_t,\quad
              \ddot q^i_{\alpha\beta}(S_q)=\left(\frac{\partial^2q^i(\gamma(t))}{\partial t^\alpha\partial t^\beta}\right)_t.$$

\medskip

From a $\gamma$ as above,  we can define its second order prolongation
\begin{equation}\label{prolongation2}
\gamma^2\colon \U\to Q_k^2,\qquad t\mapsto \gamma^2(t):=\{\text{morphism}\,\,\A\overset{\gamma^*}\to\C^\infty(\U)/\m_t^3\simeq \R_k^2\},
\end{equation}

We can consider a different but related concept: iterated bundles of velocities: $(Q_k^1)_k^1$. The relation between $Q_k^2$ and $(Q_k^1)_k^1$ generalizes that existent between the second order tangent bundle $T^2Q$ and the iterated tangent bundle $T(TQ)$ (case $k=1$).

\begin{prop}\label{Taylor}
There exist a canonical immersion $Q_k^2\subset (Q_k^1)_k^1$. In coordinates $\{q^i,\dot q^i_\alpha,\ddot q^i_{\alpha\beta}\}$ on $Q_k^2$ and  $\{q^i,\dot q^i_\beta, {(q^i)}^{\centerdot}_\alpha,{({\dot q}^i_\beta)}^{\centerdot}_\alpha\}$ on $(Q_k^1)_k^1$, that immersion is described by the equations
$${(q^i)}^{\centerdot}_\alpha=\dot q^i_\alpha,\quad {({\dot q}^i_\beta)}^{\centerdot}_\alpha=\ddot q^i_{\alpha\beta}.$$

\end{prop}

 The proof, valid for a more general situation, can be found in \cite{MMR1} according the following steps: it is a local question and, when $Q=\R^n$, it is easy to see that $Q^A=Q\otimes A$ for any Weil algebra $A$; as a consequence, if $B$ is another Weil algebra, we have $(Q^A)^B=Q^{A\otimes B}$ for $Q=\R^n$ and, so, also for an arbitrary manifold $Q$ (this is a result in \cite{Weil}); now, for $A=B=\R_k^1$ we derive $(Q_k^1)_k^1=Q^{\R_k^1\otimes\R_k^1}$; considering the morphism of algebras
$\R_k^2\to \R_k^1\otimes\R_k^1$ defined by the rule $\mu\colon\epsilon^\alpha\mapsto \epsilon^\alpha\otimes 1+1\otimes \epsilon^\alpha$ we get the required smooth map
$Q_k^2\to (Q_k^1)_k^1$  as follows: given $S_q\in Q_k^2$ we consider the composition
$$\A\overset{S_q}\to\R_k^2\overset{\mu}\to\R_k^1\otimes\R_k^1$$
which defines a point of $(Q_k^1)_k^1$.

\begin{obs}
The theory of Weil bundles is strongly related to that of \emph{jet bundles}: the points of a Weil bundle $Q_k^r$  are, by definition, homomorphisms $\A\to \R_k^\ell$. Those that are epimorphisms constitute an open subset $\check Q_k^\ell\subset Q_k^\ell$. Taking the kernels of the homomorphisms in $\check Q_k^\ell$, we obtain the bundle of jets of order $r$ of $k$-dimensional submanifolds of $Q$, $J_k^r(Q)$, as a set of ideals of $\A$ (jets of sections of a fiber bundle $Q\to M$ are recovered as certain open subset of $J_k^\ell(Q)$); this point of view was developed in  \cite{MMR1} and entails a number of advantages; as an example of application of this approach, see \cite{MMRinvariantes}). The theory can be generalized to Weil bundles for any Weil algebra  $A$ as demonstrated in \cite{A1}, where it was showed that $\check Q^A\to J^A(Q)$ is a principal fiber bundle; even the notion of contact system was extended to these spaces \cite{AM}. Finally, a general version of jets understood as a part of the spectrum of the ring $\A$ was presented in \cite{Primary spectrum}.
\end{obs}

\bigskip

\section{Second order partial differential equations. The geodesic $k$-field}\label{S:geodesic}

In this section we will generalize the concept of second order differential equation (Definition \ref{def: second_order_equation}) to the $k$-parametric  case.

\begin{defi} A section $D\colon Q_k^1\to Q_k^2$ (of the natural projection $Q_k^2\to Q_k^1$) is called a \emph{second order partial differential equation}. In local coordinates, we have
   $$D\colon (q^i,\dot q^i_\alpha)\mapsto (q^i,\dot q^i_\alpha,\ddot q^i_{\alpha\beta}=A_{\alpha\beta}^i(q,\dot q))$$
   for suitable functions $A_{\alpha\beta}^i\in\C^\infty(Q_k^1)$.
  Moreover, if we identify $Q_k^2$ as submanifold of the iterated $(Q_k^1)_k^1$ (Proposition \ref{Taylor}), then $D$ is a $k$-tuple of vectors fields $(D_1,\dots,D_k)$ on $Q_k^1$; in local coordinates,
  $$D_\alpha=\dot q^i_\alpha\,\frac{\partial}{\partial q^i}+A^i_{\alpha\beta}(q,\dot q)\,\frac{\partial}{\partial \dot q^i_\beta}.$$
   \end{defi}

  In other words, $D=(D_1,\dots,D_k)$ can be characterized by the following property,
  \begin{equation}\label{SOPDEpropiedad}
  \forall f\in\A,\qquad D_\alpha f=\dot\partial_\alpha f,\quad\text{and}\quad [D_\alpha,D_\beta]f=0.
  \end{equation}
  When $D$ is considered as a field of $k$-velocities on $Q_k^1$, defines a system of first order partial differential equations on $Q_k^1$ (see Definition \ref{FOPDE}). In local coordinates (as in Equation (\ref{FOPDElocal})) we get the system:
  \begin{equation}
  \begin{cases}
  {\displaystyle\frac{\partial q^i(\gamma(t))}{\partial t^\alpha}}=\dot q^i_\alpha(\gamma(t))\\[0.8em]
  {\displaystyle\frac{\partial \dot q^i_\beta(\gamma(t))}{\partial t^\alpha}}=A_{\alpha\beta}^i(q^j(\gamma(t)),\dot q^j_\nu(\gamma(t))
  \end{cases}
  \end{equation}
  \vskip .3cm
 \noindent (with $t=(t^1,\dots,t^k)$) which turns to be the following system of second order partial differential equations:
 \begin{equation}\label{SOPDElocal}
 \frac{\partial^2\,  q^i(t)}{\partial t^\alpha\partial t^\beta}=A_{\alpha\beta}^i(q^j(t),\partial q^j(t)/\partial t^\nu),\quad i=1,\dots,n\, ;\alpha,\beta=1,\dots k
 \end{equation}
   (where, in order to lighten the notation, we have put $q^i(t)$ instead of $q^i(\gamma(t))$).
 \medskip

 \subsection{Geodesic $k$-field}
 It will be essential in the sequel the concept of differential equation describing  `inertial motion' or `evolution in the absence of forces' (Newton's first law):
\begin{thm}[Geodesic $k$-field]
Each given metric $g$ on $Q$ determines canonically a second order partial differential equation $D_G$ which will be called \emph{geodesic $k$-field}.
\end{thm}
\begin{proof}
Let us denote by $\textrm{exp}$ the exponential map associated with $g$. Given a $X_q=(X_{1},\dots,X_{k})\in Q_k^1$,  we consider the smooth map
$$\gamma\colon\R^k\to Q,\qquad (t^1,\dots,t^k)\mapsto\textrm{exp}(t^1X_{1}+\cdots+t^kX_{k}),$$
obtained by composing ${X_q}\colon\R^k\to T_qQ$ with the exponential map  ${\textrm{exp}}\colon T_qQ\to Q$ (just defined in a neighborhood of $0\in T_qQ$).
Then we can define $D_G\colon Q_k^1\to Q_k^2$ by the rule
$$D_G(X_q):=\gamma^2(0),$$
where $\gamma^2$ is the second prolongation (\ref{prolongation2}) of $\gamma$.
\end{proof}
\medskip

A straightforward computation gives $D_G=(D_1,\dots,D_k)$ with
$$D_{\alpha}=\dot q^i_\alpha\frac{\partial}{\partial q^i}-\Gamma_{jk}^i\dot q^j_\alpha \dot q^k_{\beta}\frac{\partial}{\partial \dot q^i_{\beta}}$$
and the corresponding system of second order partial differential equations:
\begin{equation}\label{SOPDEgeodesica}
 \frac{\partial^2\,  q^i(t)}{\partial t^\alpha\partial t^\beta}=-\Gamma_{jk}^i(q(t))\,
   \frac{\partial q^j(t)}{\partial t^\alpha}\,\frac{\partial q^k(t)}{\partial t^\beta}.
  \end{equation}
\bigskip

\section{The dual: $k$-symplectic canonical structure. The Liouville form}\label{S:Liouville}

The bundle $Q_k^1\to Q$ possesses  a dual $(Q_k^1)^*$, defined as the set of linear maps $T_qQ\to\R^k$,  $q\in Q$. These elements, that will be called \emph{$(k,1)$-covelocities}, correspond to tangent maps of smooth maps from (a neighborhood of $q$ in) $Q$ to $\R^k$. Equivalently, the elements of $(Q_k^1)^*$ are tuples of $k$ differential 1-forms $\sigma_q=(\sigma^1,\dots,\sigma^k)$, $\sigma^\alpha\in T^*_qQ$ or also $\sigma=\sigma^\alpha\otimes e_\alpha$ where $\{e_\alpha\}$ is the canonical basis for $\R^k$. We can endow $(Q_k^1)^*$ with a differentiable structure in the following way: if $\{q^i\}$ is a local chart on $Q$, each covelocity $\sigma_q\in (Q_k^1)^*$ can be expressed as
 $$\sigma_q=p_i^\alpha\,d_qq^i\otimes e_\alpha,$$
 for appropriate scalars $p_i^\alpha\in\R$. The sets of functions $\{q^i,p^\alpha_i\}$ so defined, provide us with  local charts on $(Q_k^1)^*$.

 In local coordinates the duality or coupling of elements $X_q=\dot q^i_\alpha(\partial/\partial q^i)_q\otimes e^\alpha\in Q_k^1$ and
    $\sigma_q=p_i^\alpha\,d_qq^i\otimes e_\alpha\in (Q_k^1)^*$ is given by
    $$\langle \sigma_q,X_q\rangle=p_i^\alpha\dot q_\alpha^i.$$

 On the other hand, it is also well defined the coupling with values in the endomorphisms of $\R^k$, that we denote as an interior product:
 \begin{equation}\label{contraccion}
\iiota_{X_q}\sigma_q\colon\R^k\to\R^k,\quad \lambda\mapsto \sigma_q(X_q(\lambda)).
 \end{equation}
 In local coordinates,
 $$i_{X_q}\sigma_q=p_i^\alpha\dot q_\beta^i\,\, e^\beta\otimes e_\alpha$$
 If we denote by $\textrm{tr}$ the trace map $\textrm{End}(\R^k)\to\R$, then
 $$\langle \sigma_q,X_q\rangle=\textrm{tr}\left(i_{X_q}\sigma_q\right).$$

\subsection{Liouville form}
As in the case of the cotangent bundle, $(Q_k^1)^*$ is endowed of a tautological structure: at each $\sigma_q\in (Q_k^1)^*$ we can define a $k$-tuple of differential 1-forms $\theta_{\sigma_q}=(\theta^1,\dots,\theta^k)$ (a $(k,1)$ covelocity on $(Q_k^1)^*$) by the rule
$$\theta^\alpha(D_{\sigma_q}):=\sigma^\alpha(\pi_*D_{\sigma_q}),$$
where $\pi\colon (Q_k^1)^*\to Q$ is the canonical projection. We call $\theta$ the (generalized) \emph{Liouville form} and, in local coordinates,
\begin{equation}\label{Liouvillelocal}
\theta^\alpha=p^\alpha_i\,dq^i\qquad\text{or}\qquad \theta=\theta^\alpha\otimes e_\alpha=p^\alpha_i\,dq^i\otimes e_\alpha.
\end{equation}

\begin{obs}\label{calculovalorado}
When we work with differential forms, vector fields or, more in general, tensor fields with values in $\R^k$, $(\R^k)^*$ or $\textrm{End}\,\R^k$, we can perform `differential operations' (exterior differential, Lie derivatives, etc.) componentwise: as an example, given a basis $\{e_\alpha\}$ of $\R^k$, if $\nu=\nu^\alpha\otimes e_\alpha$, where $\nu^\alpha$ are ordinary differential forms, then it is well defined its exterior differential $d\nu:=(d\nu^\alpha)\otimes e_\alpha$.
\end{obs}

The exterior differential of the Liouville form $\theta$ defines the so called \emph{polysymplectic} structure of $(Q_k^1)^*$:
\begin{equation}\label{symplectic}
d\theta=d \theta^\alpha\otimes e_\alpha
\end{equation}
or, in local coordinates,
$$d\theta=dp^\alpha_i\wedge dq^i\otimes e_\alpha.$$

\medskip

\subsection{Isomorphism induced by a metric}
When $Q$ is endowed with a pseudoriemannian metric $g$ (non singular symmetric 2-covariant tensor field of arbitrary signature), we can translate the structures of the velocity bundle to the covelocities one and conversely. Specifically, given $X_q\in Q_k^1$ we can define $\sigma_q\in (Q_k^1)^*$ to be the transposed
of the composition:
$$\R^k\overset{X_q}\To T_qQ\overset{g}\simeq T^*_qQ$$
followed with the identification $(\R^k)^*\simeq \R^k$, $e^\alpha\to e_\alpha$ (via the canonical basis or the canonical euclidean metric).
Hence, we obtain an isomorphism of bundles on $Q$,
\begin{equation}\label{isomorfismometrico}
Q_k^1\simeq (Q_k^1)^*
\end{equation}
which in coordinates is written as
\begin{equation}\label{isomorfismometrico2}
p^\alpha_i:=g_{ij}\,\dot q^{j\alpha} 
\end{equation}
if $g=g_{ij}\,dq^i\,dq^j$.

\begin{obs}
The raising and lowering of indices $\alpha$ is made by means of the canonical euclidean metric on $\R^k$; as an example,
$\dot q^i_\alpha=\dot q^{i\beta}\delta_{\alpha\beta}$, where $\delta_{\alpha\beta}$ is the Kronecker delta.
\end{obs}
\medskip

From now on, we suppose $g$ is given and  we work indistinctly  on $Q_k^1$ or $(Q_k^1)^*$  under the isomorphism (\ref{isomorfismometrico}). For instance, we have a Liouville form on $Q_k^1$ and second order partial differential equations on $(Q_k^1)^*$.

\begin{obs}
The isomorphism induced by $g$, plays the role of the Legendre transformation in \cite{Gunther}. In fact, in a sense, it is a particular case.
\end{obs}

\bigskip
\section{The generalized Newton's second law}\label{S:Newton}

With all necessary tools previously introduced, in this section we will get a $k$-dimensional version of the Newton's second law.
\smallskip

Let $g$ a (non singular) metric on $Q$ of arbitrary signature. By means  of $g$ we translate the (generalized) Liouville form $\theta$ to $Q_k^1$ and define $\Tau$: the differential 1-form with values in the matrixes $\textrm{End}\,\R^k$ determined by
\begin{equation}\label{energiacinetica1}
i_{D_G}\,d\theta+\Tau=0\quad\text{or}\quad \Tau=-\left(i_{D_\alpha}\,d\theta^\beta\right)\otimes e_\beta\otimes e^\alpha
\end{equation}

In a local chart, a straightforward by lengthy computation shows the following coordinate expression:
$\Tau=\Tau_\alpha^\beta\otimes e_\beta\otimes e^\alpha$ with
\begin{equation}\label{kenergiacinetica}
\Tau_\alpha^\beta=-i_{D_{G\alpha}}d\theta^\beta=d\left(\frac 12 g_{jh}\dot q^j_\alpha\dot q^{h\beta}\right)+
          \frac 12\left(\frac{\partial g_{ij}}{\partial q^h}-\frac{\partial g_{ih}}{\partial q^j}\right)\dot q^j_\alpha \dot q^{h\beta}\,dq^i
          +\frac 12 g_{jh}\left(\dot q^j_\alpha\,d\dot q^{h\beta}-\dot q^{j\beta}\,d\dot q^h_\alpha\right).
\end{equation}

  \begin{lemma}\label{diferenciaSODE}
  Let $D,D'\colon Q_k^1\to Q_k^2$ be two second order partial differential equations. Then, their difference $V\colon Q_k^1\to Q_k^2\subset (Q_k^1)_k^1$ is vertical with respect to the tangent map of the projection $Q_k^1\to Q$. Conversely, the sum of a second order differential equation $D$ and a vertical $V$ gives a new second order differential equation.
    \end{lemma}
  \begin{proof}
  It follows from (\ref{SOPDEpropiedad}) that $(D_\alpha-D'_\alpha)f=\dot\partial_\alpha f-\dot\partial_\alpha f=0,\quad\forall f\in\A.$
  \end{proof}

 In local coordinates, $V$ is vertical if $V=V_\alpha\otimes e^\alpha$, with $V_\alpha=V_{\alpha\beta}^i\,\partial/\partial\dot q^i_\beta$ for certain coefficients $V_{\alpha\beta}^i$ (symmetric in $\alpha\beta$).

 A differential 1-form on $Q_k^1$ is said to be \emph{horizontal} if vanishes when is applied to a vertical vector (as before, we mean vertical with respect to the projection $Q_k^1\to Q$). In local coordinates a form is horizontal if locally is a linear combination of $dq^i$, $i=1,\dots,n$, with coefficients in $\C^\infty(Q_k^1)$.

  \begin{lemma}\label{vertical}
  For each vertical section $V\colon Q_k^1\to Q_k^2$, the inner contraction $i_Vd\theta$ is a horizontal 1-form valued in $\textrm{End}\,\R^k$. Conversely, if $i_Vd\theta$ is horizontal, then $V$ is vertical.
  \end{lemma}
  \begin{proof}
  It is enough to take local coordinates: by (\ref{isomorfismometrico2}), $V_\alpha p_i^\beta=g_{ij}V_\alpha^{i\beta}$. Hence, if $V$ is vertical, then we get
  $\iiota_{V_\alpha}d\theta^\beta=g_{ij}V_\alpha^{i\beta}\,dq^j$ which is a horizontal form. Conversely, if $\iiota_{V_\alpha}d\theta^\beta$ are horizontal, then $V_\alpha q^i=0$, so that $V$ is vertical.
  \end{proof}

  From the above lemmas we get the
  \begin{thm}[Newton's second  law for fields]
  With the same notation. Each metric $g$ on $Q$ determines a univocal correspondence between horizontal 1-forms $F$ valued in $\textrm{End}\,\R^k$ and second order differential equations $D\colon Q_k^1\to Q_k^2$ given by the rule
  \begin{equation}\label{Newton1}
 \iiota_D\,d\theta+\Tau=F
  \end{equation}
  \end{thm}
  \begin{proof}
  I follows from Lemmas \ref{diferenciaSODE} and \ref{vertical} by taking into account that
  $D-D_G$ is vertical and $i_{D_G}d\theta+\Tau=0$.
  \end{proof}

  In local coordinates, if $D=D_G+V$ with $V=V_\alpha\otimes e^\alpha$ where $V_\alpha=V_{\alpha\beta}^{i}\frac{\partial}{\partial \dot q^i_\beta}$ and if
  $F=F_{\alpha}^\beta\otimes e_\beta\otimes e^\alpha$ with $F_{\alpha}^\beta=F_{j\alpha}^\beta\,dq^j$, then
    \begin{equation}\label{VA}
    V_{\alpha\beta}^{i}=g^{ij}F_{j\alpha\beta}.
    \end{equation}

    By combining Equations (\ref{VA}), (\ref{SOPDEgeodesica}) we get the coordinate expression of this generalized Newton's law (once the `form of work' or force $F$ is given):

    \begin{equation}\label{Newtoncoordenada}
 \frac{\partial^2\,  q^i}{\partial t^\alpha\partial t^\beta}=-\Gamma_{jk}^i(q)\,
   \frac{\partial q^j}{\partial t^\alpha}\,\frac{\partial q^k}{\partial t^\beta}+g^{ij}(q)F_{j\alpha\beta}(q),\qquad \alpha,\beta=1,\dots k,
  \end{equation}
  \smallskip

  \noindent where $t=(t^1,\dots,t^k)$, $q^i=q^i(t)$, $g^{ij}$  is the $(i,j)$ entry of the inverse matrix of $(g_{ij})$, etc., as usual.
 \bigskip

\section{Canonical equations and the Hamilton-Noether principle}\label{S:Lagrange}

In order to connect the Newton's second law with the Hamiltonian or Lagrangian theory we just have to take trace:
$$\textrm{tr}\colon \textrm{End}\,\R^k\to\R, \quad \left(a_{\alpha}^\beta\right)\mapsto a_\alpha^\alpha.$$

In particular, from (\ref{kenergiacinetica}) we see that
\begin{equation}\label{trazacinetica}
\textrm{tr}\,\Tau=d\left(\frac 12 g_{jh}\dot q^j_\alpha\dot q^{h\alpha}\right)
\end{equation}
and taking the trace on (\ref{energiacinetica1}) we get:
  $$i_{(D_{G})_\alpha}\,d\theta^\alpha+dT=0,\qquad\text{where}\quad T:=\frac 12 g_{jh}\dot q^j_\alpha\dot q^{h\alpha}.$$
The function \emph{$k$-kinetic energy} $T$ can be globally defined on $Q_k^1$ as follows:
$$T:=\frac 12\sum_\alpha g(\dot\partial_\alpha,\dot\partial_\alpha).$$

For the same reason, if we take trace on the relation (\ref{Newton1}) we arrive to
\begin{equation}\label{casiLagrange}
i_{D_\alpha}d\theta^\alpha+dT=F_\alpha^\alpha
\end{equation}

In the `conservative case', say $F_\alpha^\alpha=-dU$ for certain function $U\in\A$, we get $$i_{D_\alpha}d\theta^\alpha+dH=0,$$ where
$H:=T+U$. So, for such an $H$, we obtain the classical Hamilton's canonical equations in field theory (or Hamilton-De Donder-Weyl equations; see for instance \cite{Rund}, Ch. IV, formulae (2.26)+(2.35)):
\begin{equation}\label{Hamilton}
\begin{cases}
  {D_\alpha q^i\phantom{.}=\phantom{-}\displaystyle\frac{\partial H}{\partial\, p_i^\alpha}}\\[1em]
  {D_\alpha p_i^\alpha=\displaystyle -\frac{\partial H}{\partial\, q^i}}
  \end{cases}
\end{equation}

By construction, a solution of $D$ in (\ref{Newton1}) is always a solution of (\ref{Hamilton}). However, the converse is not true. In order to illustrate this fact, we give the following elementary
\medskip

\noindent{\bf Example.}  Let us consider $Q=\R^n$ endowed with the euclidean metric $g_{ij}=\delta_{ij}$, so that Christtoffel symbols $\Gamma_{ij}^\ell$ vanish.  Consider also the null force $F=0$.
Assume that $D$ holds (\ref{Newton1}), in such a way that $D$ is the geodesic $k$-field for the flat metric: $D_\alpha=\dot q^i_\alpha\partial/ \partial_{q^i}$. Therefore, $D$ defines the following differential equations system (Newton's second law)
\begin{equation}\label{SOPDElocalejemplo}
   \displaystyle\frac{\partial^2\,  q^i(t)}{\partial t^\alpha\partial t^\beta}=0
\end{equation}
whose solutions are
\begin{equation}\label{SOPDElocalejemplosoluciones}
   q^i(t)=a^i+b_\alpha^i\,t^\alpha,\quad a^i,b_\alpha^i\in\R.
\end{equation}

If, on the other hand, we take the trace, we get:
$\textrm{tr}\,F=0$, so that, $H=\sum_{i,\alpha}(\dot q^i_\alpha/2)^2$ and then, a little calculation gives (for (\ref{Hamilton})) the Laplace equation:
\begin{equation}
\sum_\alpha\frac{\partial^2\,  q^i(t)}{\partial\, (t^\alpha)^2}=0
\end{equation}
whose solutions are
\begin{equation}\label{SOPDElocalejemplosoluciones2}
   q^i(t)=\text{harmonic functions of $t^1,\dots,t^k$}
\end{equation}

Obviously, the set (\ref{SOPDElocalejemplosoluciones}) is only a part of (\ref{SOPDElocalejemplosoluciones2}).
\bigskip

Equations (\ref{Hamilton}) show, by mere inspection, that taking the trace on the Newton's equation, we get the canonical equations. However, in the sequel we will prove this fact intrinsically by obtaining the needed variational properties, including Hamilton's principle and  Noether's theorem.
\medskip

\noindent{\bf Infinitesimal transformations.} Let us consider a tangent vector field $v=v^i\partial/\partial q^i$ on $Q$ and denote by $\{\tau_t\}_{t\in\R}$ its local uniparametric group (or flow). Then it define by means of the tangent map a new uniparametric group on $Q_k^1$, say $\{\overline\tau_t\}_t$. The vector field that generates $\overline \tau$ is, in local coordinates,
\begin{equation}\label{prolongacion1}
\delta_v=v^i\frac{\partial}{\partial q^i}+\dot v^i_\alpha \frac{\partial}{\partial \dot q^i_\alpha},\qquad\text{where}\quad \dot v_\alpha^i=\dot\partial_\alpha v^i.
\end{equation}
More intrinsically, $\delta$ is determined by the following feature: if $f\in\A$, then $\delta_v f=vf$ and
\begin{equation}\label{infinitesimal}
\delta_v\dot\partial_\alpha f=\dot\partial_\alpha \delta_v f.
\end{equation}

The vector field $\delta_v$ will be called \emph{infinitesimal transformation induced by $v$} (or also, \emph{canonical prolongation of $v$ to $Q_k^1$}).
\medskip

Now, from (\ref{casiLagrange}) it is easy to follow the steps in \cite{MecanicaMunoz} (or also, \cite{RM}) to get
\begin{thm}[Hamilton-Noether principle]
Maintaining the above notations and $\delta=\delta_v$. It holds
\begin{equation}\label{Lagrange-Noether}
D_\alpha\langle \theta^\alpha,\delta\rangle=\langle dT+F_\alpha^\alpha,\delta\rangle
\end{equation}
In the `conservative case': $F_\alpha^\alpha=-dU$ for an appropriate function $U\in\A$, we have
\begin{equation}\label{Lagrangian}
D_\alpha\langle \theta^\alpha,\delta\rangle=\delta(L),
\end{equation}
where, by definition, $L:=T-U$ (\emph{Lagrangian function}).
\end{thm}
\begin{proof}
It is sufficient to show (\ref{Lagrange-Noether}).
We have $D_\alpha\langle \theta^\alpha,\delta\rangle=\langle L_{D_\alpha}\theta^\alpha,\delta\rangle+\langle \theta^\alpha,[D_\alpha,\delta]\rangle$; but, on the one hand, from (\ref{casiLagrange}) and the identity $i_{D_\alpha}\theta^\alpha=2T$ we derive
$L_{D_\alpha}\theta^\alpha=i_{D_\alpha}d\theta^\alpha+di_{D_\alpha}\theta^\alpha=dT+F_\alpha^\alpha$ and, on the other hand, $[D_\alpha,\delta]$ is vertical: if $f\in\A$,
$$[D_\alpha,\delta]f=D_\alpha(\delta f)-\delta(D_\alpha f)=\dot\partial_\alpha(\delta f)-\delta(\dot\partial_\alpha f)=0,$$
due to (\ref{infinitesimal}).
This is why $\langle \theta^\alpha,[D_\alpha,\delta]\rangle=0$ and the statement follows.
\end{proof}
\bigskip

\noindent Now, the consequences of the above theorem are obtained as in the classical case:
\medskip

\noindent \textbf{Hamilton's principle.} Let  $\gamma\colon \U\to Q$ be a solution of the second order equation $D$: that is to say, $(\gamma_*)_t(\partial/\partial t^\alpha)_t=D_{\alpha,\gamma(t)}$. In that case, for an arbitrary infinitesimal variation $\delta=\delta_v$ we have
    $$\int_{\U} (\gamma^1)^*\delta(L)\,dt=\int_\U(\gamma^1)^*D_\alpha\langle \theta^\alpha,\delta\rangle\,dt=
         \int_{\U}\frac{\partial}{\partial t^\alpha}((\gamma^1)^*\theta^\alpha(v))dt
         =\int_{\U}d\left((\gamma^1)^*\theta^\alpha(v)\,i_{\partial/\partial t^\alpha}dt\right),
         $$
         where $dt:=dt^1\wedge\dots\wedge dt^k$.
         By Stokes' theorem, the last integral equals to $\int_{\partial\,\U}(\gamma^1)^*\theta^\alpha(v)\,i_{\partial/\partial t^\alpha}dt$. Hence, when the variation $v$ vanishes along the boundary $\partial\, \U$ we get the Hamilton principle
  $$\int_{\U} (\gamma^1)^*\delta(L)\,dt=0.$$
  \medskip

 \noindent{\bf Noether conservation laws.} Let us suppose that $v$ is a symmetry of the Lagrangian function $L$; that is to say, $\delta L=0$. Let us take an arbitrary  solution  $\gamma\colon \U\to Q$  of  $D$. Then, substitution in  (\ref{Lagrangian}) and pull-back by $\gamma^1$ implies
    $$\frac{\partial}{\partial t^\alpha}((\gamma^1)^*\theta^\alpha(v))=0,$$
 which defines the \emph{conservation law} associated to the symmetry $v$.

\bigskip


\end{document}